\documentclass[10pt, conference]{IEEEtran}
\IEEEoverridecommandlockouts
\usepackage{amsmath}
\usepackage{amssymb}
\usepackage{url}
\usepackage{cite}
\usepackage{flushend}
\usepackage{makecell}
\usepackage{subfigure}
\usepackage{amsfonts,amsthm}

\usepackage{algorithmic}
\usepackage{algorithm}

\usepackage{graphicx}
\usepackage{textcomp}
\usepackage{float}
\usepackage{booktabs}
\usepackage{bm}
\usepackage{float}  
\usepackage{multirow}
\usepackage[normalem]{ulem}
\useunder{\uline}{\ul}{}
\usepackage{xcolor}
\def\BibTeX{{\rm B\kern-.05em{\sc i\kern-.025em b}\kern-.08em
    T\kern-.1667em\lower.7ex\hbox{E}\kern-.125emX}}
\allowdisplaybreaks
\newtheorem{lemma}{Lemma}
\newtheorem{theorem}{Theorem}
\newtheorem{corollary}{Corollary}

\newtheorem{Assumption}{Assumption}

\begin{document}

\title{Robust Event-Triggered Integrated Communication and Control with Graph Information Bottleneck Optimization}

\author{Ziqiong Wang, Xiaoxue Yu, Rongpeng Li, and Zhifeng Zhao
    \thanks{This work was supported in part by the National Key Research and Development Program of China under Grant 2024YFE0200600, in part by the Zhejiang Provincial Natural Science Foundation of China under Grant LR23F010005.}
    
    \thanks{Z. Wang, X. Yu and R. Li are with the College of Information Science and Electronic Engineering, Zhejiang University (email: \{wangziqiong, sdwhyxx, lirongpeng\}@zju.edu.cn).}
    \thanks{Z. Zhao is with Zhejiang Lab as well as the College of Information Science and Electronic Engineering, Zhejiang University (email: zhaozf@zhejianglab.com).}
}

\maketitle
   
\begin{abstract}
Integrated communication and control serves as a critical ingredient in Multi-Agent Reinforcement Learning. However, partial observability limitations will impair collaboration effectiveness, and a potential solution is to establish consensus through well-calibrated latent variables obtained from neighboring agents. Nevertheless, the rigid transmission of less informative content can still result in redundant information exchanges. Therefore, we propose a Consensus-Driven Event-Based Graph Information Bottleneck (CDE-GIB) method, which integrates the communication graph and information flow through a GIB regularizer to extract more concise message representations while avoiding the high computational complexity of inner-loop operations. To further minimize the communication volume required for establishing consensus during interactions, we also develop a variable-threshold event-triggering mechanism. By simultaneously considering historical data and current observations, this mechanism capably evaluates the importance of information to determine whether an event should be triggered.
Experimental results demonstrate that our proposed method outperforms existing state-of-the-art methods in terms of both efficiency and adaptability.
\end{abstract}

\begin{IEEEkeywords}
Communication and control co-design, event trigger, graph information bottleneck optimization, consensus-oriented, multi-agent reinforcement learning. 
\end{IEEEkeywords}

\section{Introduction}\label{sec1_Introduction}

Nowadays, the thriving development of Multi-Agent Systems (MAS) \cite{MAS} has propelled the integration of communication and control into a pivotal research direction, showcasing significant system-wide advancements. 
Typically, such a co-design is contingent on Multi-Agent Reinforcement Learning (MARL) \cite{MARLsurvey}, but suffers from
partial observability, as agents can only access limited or incomplete information about the environment and other agents' states. 
Therefore, in order to make informed collaborative decisions, it necessitates the information sharing and aggregation among agents \cite{Tarmac2019}. 
Although full-mesh, raw data transmission \cite{masia2022} could theoretically resolve this problem, the overwhelming information exchange makes it often practically infeasible, especially under bandwidth-limited scenarios.
Until recently,
there starts to emerge algorithms \cite{xiang2023decentralized} that establish consensus over well-calibrated latent variables from neighbors only. 
Compared to direct data transmission, the consensus-based algorithm boosts scalable decentralized execution.
Nevertheless, the insufficiency of communication and control co-design, such as the blunt transmission of all less informatively transmitted messages \cite{niu2021}, leaves significant redundancy in exchanged messages. Correspondingly, the inevitable difficulties posed by limited bandwidth and noisy channels still call for a more efficient algorithmic framework. 

As a remedy, theoretically guided refinement and compression contribute to squeezing the amount of exchanged information for consensus inference. In that regard, the integration of Information Bottleneck (IB) emerges as a highly promising direction to improve overall communication efficiency. In the context of general representation learning, the IB principle \cite{IB2020} emphasizes that the optimal representation should contain sufficient and minimal information that is beneficial for ultimate tasks. Concurrently, TOCF \cite{TOCF} introduces IB theory into the multi-agent communication reinforcement learning (MACRL) scenarios, enabling efficient message compression in communication. However, such IB-based representation learning methods typically require input data to meet independent and identically distributed (i.i.d.) conditions, which do not always hold in the context of multi-agent communication. Meanwhile, Ref. \cite{TOCF} overlooks that both graph structure and agent features carry important information in MACRL, whereas Graph Information Bottleneck (GIB) theory \cite{GIB} introduces a local-dependence assumption and provides a paradigm that regularizes the topological as well as the feature attributes, offering significant advantages, but has not yet been applied to Reinforcement Learning (RL). In this regard, MAGI \cite{MAGI2024} extends the GIB principle to MACRL methods to derive more effective and concise message representations. However, the classical MAGI framework optimizes the communication graph and information flow separately. Furthermore, MAGI requires agents to communicate with all agents within a certain range before assessing the importance of information. Therefore, it inevitably adds significantly to the communication volume, necessitating the design of a more efficient information compression approach.

Conventional control algorithms have adopted event-triggered mechanisms \cite{kim2019learning,ATOC2018}, which reduce redundant exchanged information by appropriately tuning the frequency of communications. However, these methods often determine the triggering moment solely based on current observations, but neglect the influence of historical information. In contrast, ETCNet \cite{ETCNET2023} enables agents to make more informed decisions by leveraging both immediate and historical data. Nevertheless, it fails to offer a detailed explanation regarding the criteria for setting the threshold, while many articles \cite{He2022} rely on fixed thresholds, potentially encountering accumulative errors due to stale updates. Furthermore, the effective incorporation of event-triggered communications into consensus inference still awaits for comprehensive investigation. 

In this paper, we propose the Consensus-Driven Event-based GIB (CDE-GIB) algorithm. Specifically, 
in contrast to the separate compression of the communication graph and data flow \cite{MAGI2024,GIB}, we propose a more efficient GIB method for joint optimization, thus avoiding high computational complexity due to inner loops. Additionally, to further reduce the communication volume required for agents to reach consensus during interaction, we introduce a variable-threshold event-triggering mechanism that takes account of both historical data and current observations and determines whether to trigger an event from the perspective of information importance.
In comparison to existing works in the literature, the contribution of this paper can be summarized as follows.

\begin{itemize}
    \item We propose the CDE-GIB framework, which novelly combines event-triggering and GIB, to ameliorate the message inefficiency of consensus inference in decentralized MARL.
    \item We develop a variable-threshold event-triggering mechanism (VT-ETM) that dynamically evaluates the information importance towards inferring the consensus across multiple agents. Such an information importance-driven event-triggering mechanism also significantly distinguishes with existing MARL-empowered solutions \cite{kim2019learning,ATOC2018,ETCNET2023}.
    \item Additionally, we devise a GIB regularizer that fuses the communication graph and information flow to obtain more concise message representations, which are applicable in consensus inference in decentralized MARL, thereby improving the efficiency of downstream control tasks.
    \item We validate the universal effectiveness of our framework through extensive simulations in the multi-agent particle environment \cite{MPE2017}.
\end{itemize}

The remainder of the paper is organized as follows. Sec. \ref{sec2_System Model} briefly introduces the system model and formulates the problem. Sec. \ref{sec3_Framework} presents the overview of our proposed CDE-GIB framework. In Sec. \ref{sec4_Experiment}, we elaborate on the experimental results and discussions. Finally, Sec. \ref{sec5_Conclusions} concludes the paper.

\section{System Model and Problem Formulation}\label{sec2_System Model}

\begin{figure}[tp]
    \centering
    \includegraphics[scale=0.4]{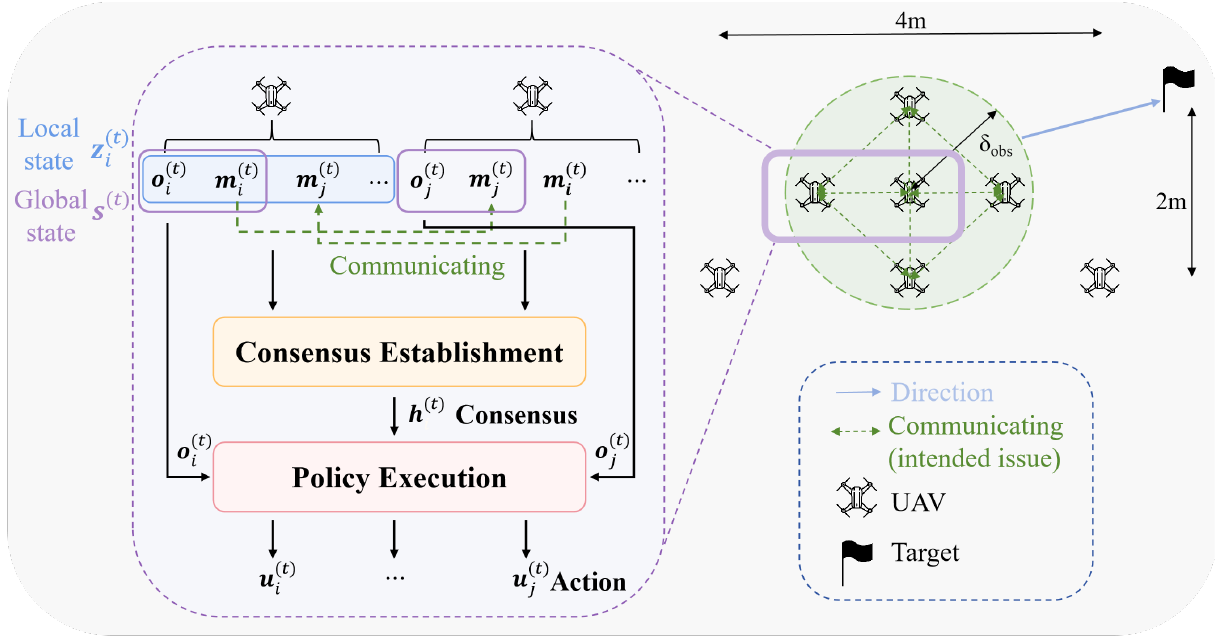}
    \vspace{-0.2cm}
\caption{Illustration of MARL information control.} 
\vspace{-1em}
\label{environment}
\end{figure} 

Beforehand, we summarize the main notations in Table \ref{nota}.

\begin{table}[tbp]
    \renewcommand\arraystretch{1.3}
    \centering
    \caption{Major notations used in the paper.}
    \vspace{-0.3em}
    \begin{tabular}{lm{6.5cm}}
    \hline
    \textbf{Notation} & \textbf{Definition} \\ \hline
    $\textbf{s}^{(t)},\textbf{o}_g^{(t)}$ & Global state and observation.\\
    $\textbf{z}_i^{(t)},\textbf{u}_i^{(t)}$ & Local state, individual action of agent $i$ at time step $t$.\\
    $\textbf{o}_i^{(t)},\textbf{m}_i^{(t)}$ & Local observation, exchanged message of agent $i$ at time step $t$. \\
    $\textbf{p}_i^{(t)},\textbf{v}_i^{(t)}$ & Velocity and position that constitute the local observation of agent $i$ at time step $t$.\\
    $\xi _{i}^{\left ( t \right )}$ & Neighbors within the observation range of agent $i$ at time step $t$. \\
    $\beta_i^{(t)}$ & Importance ratio of MAPPO of agent $i$ at time step $t$. \\
    $\alpha$ & Coefficient of policy entropy.\\
    $\pi_{\theta_i}$, $\theta_i$ & Target policy and its parameter of agent $i$.\\
    $\pi_{\theta_{i,old}}$, $\theta_{i,old}$ & Behavior policy and its parameter of agent $i$. \\
    $\mathcal{T}_i$ & Sequence of event-triggering times of agent $i$. \\
    $\mathcal{G}$ & Graph in the Graph Neural Network.\\
    
    \hline
    \end{tabular}
    \label{nota}
    \vspace{-0.5cm}
\end{table}
\subsection{System Model}

The MARL problems can be typically modeled as a Decentralized Partially Observable Markov Decision Process (Dec-POMDP), which is characterized by a tuple $\left \langle \mathcal{N, S, U, P, R}, \Omega, \mathcal{O, \gamma} \right \rangle $. 
$\mathcal{N}$ represents the set of $N$ active nodes.
Due to the practical communication limitation, the heterogeneous connectivity across nodes can be characterized by an adjacency matrix $\mathbf{A} \in \mathbb{R}^{N \times N}$ \texttt{--} $A(i,j) = 1$ if and only if the Euclidean distance between $i \in \mathcal{N}$ and $j \in \mathcal{N}$ is less than the maximum observation range $\delta _\text{com}$; and it nulls otherwise. Naturally, the MAS can be denoted as a graph $\mathcal{G} = (\mathcal{N}, \mathcal{E}, \mathcal{H})$, where $\mathcal{E} \subseteq \mathcal{N} \times \mathcal{N}$ represents the edge set characterized by $\mathbf{A}$, and $\mathcal{H} = \left\{ \mathbf{h}^{(t)} \mid t = 1, 2, \dots, T \right\}$ contains the feature attributes of the nodes.
$\mathcal{S}$ denotes the global state space of the problem and $\mathcal{U}$ is the homogeneous action space for the multi-agent system. The joint action $\textbf{u}^{(t)}=\left\{\textbf{u}_i^{(t)}\mid\forall i \in \mathcal{N}\right\}$ executed at the current state $\textbf{s}^{(t)}$ makes the environment transit to the next state $\textbf{{s}}^{(t+1)}$ according to the transition probability function $ \mathcal{P}(\textbf{s}^{(t+1)} \mid \textbf{s}^{(t)}, \textbf{u}^{(t)}) : \mathcal{S} \times \mathcal{U} \times \mathcal{S} \to [0, 1]$. Due to the limited capacity for perception of the complex environment, each agent $i$ acquires a local observation $\textbf{o}^{(t)}_{i} \in \Omega$ via the observation function $\mathcal{O}(\textbf{o}_{i}^{(t)} \mid \textbf{s}_i^{(t)}, i) : \mathcal{S} \times \mathcal{N} \times \Omega \rightarrow [0, 1]$. All agents share a global reward function $ \mathcal{R}(\textbf{s}^{(t)}, \textbf{u}^{(t)}) : \mathcal{S} \times \mathcal{U} \ \to \mathbb{R} $ and the overall objective is to maximize the total discounted reward $\mathbb{E}\left [  {\textstyle \sum_{t}\gamma ^{t}\mathcal{R}^{\left ( t \right ) }  }  \right ]  $, where $\gamma \in \left [ 0, 1 \right ] $ means a discount factor. In alignment with the Dec-POMDP framework, we specify the elements as follows.

1) State: The state $\mathbf{s}^{(t)} \in \mathcal{S}$ can be task-dependent. For example, in an Unarmed Aerial Vehicle (UAV) scenario, an agent $i$ can obtain direct observation $\textbf{o}_{i} ^{\left ( t \right ) }  = \left \{ \left ( \textbf{p}_{j}^{\left ( t \right ) }, \textbf{v}_{j}^{\left ( t \right ) }     \right )\mid \forall j \in \xi _{i}^{\left ( t \right ) }   \right \} $ composed of the positions $\textbf{p}_{j}^{\left ( t \right )}$ and velocities $\textbf{v}_{j}^{\left ( t \right )}$ of itself and its neighbors and receives exchanged messages $\left \{ \textbf{m}_{j}^{\left ( t \right ) }\mid \forall j \in \xi _{i}^{\left ( t \right ) }   \right \} $, where $\textbf{m}_{j}^{\left ( t \right ) }$ represents a learnable vector intended for communication and $\xi _{i}^{\left ( t \right )}$ comprises agent $i$ and its neighbors. Moreover, the global observation $\textbf{o}_g^{(t)}$ can be expressed as $\textbf{o}_g^{(t)} = \left\{ \left(\textbf{p}_i^{(t)}, \textbf{v}_i^{(t)}\right) \mid \forall i \in \mathcal{N} \right\}$. Therefore, the global state $\textbf{s}^{(t)} = \left ( \textbf{o}_g^{(t)}, \left\{ \textbf{m}_i^{(t)} \mid \forall i \in \mathcal{N} \right\}  \right ) \in \mathcal{S}$ encompasses local states $\textbf{z}_i^{(t)} = \left(\textbf{o}_i^{(t)}, \left\{\textbf{m}_j^{(t)} \mid j \in \xi_i^{(t)}\right\}\right)$ of all agents.

2) Action: Based on the local state $\textbf{z}_i^{(t)}$, each agent determines its acceleration $\textbf{u}_i^{(t)} = \left( u_{x_i}^{(t)}, u_{y_i}^{(t)} \right) \in \mathcal{U}$ according to its policy $\pi_{\theta_i}\left(\cdot \mid \textbf{z}_i^{(t)}\right)$ individually to accomplish assigned relative tasks.

3) Reward: We define the reward function as a weighted sum of multiple task-oriented and/or communication-related components, which are detailed in Section \ref{sec4_Experiment}.

\subsection{Multi-Agent Proximal Policy Optimization}\label{mappo}

To steer all agents toward maximizing the discounted accumulated reward $\mathbb{E}\left [ \sum_{t}\gamma ^{t}\mathcal{R}^{(t)}\right ]$, MAPPO \cite{MAPPO2022} is employed as the base RL method, which combines single-agent PPO \cite{PPO2017} with the centralized training and decentralized execution (CTDE) paradigm, aiming to learn both the individual policy $\pi_{\theta_i}\left(\cdot \mid \textbf{z}_i^{(t)}\right)$ for each agent $i$ and the value function $V_{\phi}(\textbf{s}^{(t)}) = \mathbb{E}_\pi \left [ \sum_{t}\gamma ^{t}\mathcal{R}^{(t)}|  \textbf{s}^{(t)}\right ]: \mathcal{S} \rightarrow \mathbb{R}$, parameterized by $\theta _{i} $ and $\phi $, respectively. Following the design of PPO, MAPPO retains the old versions of $\theta_{i,\text{old}} $ and $\phi_\text{old}$, while $\theta_{i,\text{old}} $ is used to interact with the environment and accumulate the samples. Additionally, the parameters $\theta _{i} $ and $\phi $ are periodically updated to maximize the objective function,
\begin{align}
J_{\pi_i}^{(t)}(\theta_i) &= \min\left(\beta_i^{(t)} \hat{A}^{(t)}, \ \text{clip}\left(\beta_i^{(t)}, 1 - \varepsilon, 1 + \varepsilon\right) \hat{A}^{(t)}\right), \notag \\
J_V^{(t)}(\phi) &= -\left(V_{\phi}(\textbf{s}^{(t)}) - \left(\hat{A}^{(t)} + V_{\phi_{\text{old}}}(\textbf{s}^{(t)})\right)\right)^2 \label{MAPPO},
\end{align}
where $\beta_i^{(t)} = \frac{\pi_{\theta_i}(\textbf{a}_i^{(t)} \mid \textbf{z}_i^{(t)})}{\pi_{\theta_{i,\text{old}}}(\textbf{a}_i^{(t)} \mid \textbf{z}_i^{(t)})}$ represents the importance ratio, $\varepsilon $ denotes a hyperparameter, while the Generalized Advantage Estimation (GAE) $\hat{A}^{(t)} = \sum_{l=0}^{T-t-1} (\gamma \lambda)^l \delta^{(t+l)}$ with the advantage estimate $\delta^{(t)}  = \mathcal{R}^{(t)} + \gamma V_{\phi_{\text{old}}}(\textbf{s}^{(t+1)}) - V_{\phi_{\text{old}}}(\textbf{s}^{(t)})$. Consequently, the final optimization objective of MAPPO is given by
\begin{equation}
\vspace{-0.3em}
J_{\text{MAPPO}} = \mathbb{E}_{i, t} \left[ J_{\pi_i}^{(t)}(\theta_i) + J_{V}^{(t)}(\phi) + \alpha H\left(\pi_{\theta_i}(\cdot \mid \textbf{z}_i^{(t)})\right) \right],
\end{equation}
where $\alpha$ is a coefficient and $H$ represents the entropy function.

\subsection{Problem Formulation}

As illustrated in Fig. \ref{environment}, our objective is to enable MARL-driven agents to maximize the global reward. Nevertheless, due to the partial observability of the global state $\textbf{s}^{(t)}$ and the distinction over local states $\textbf{z}_i^{(t)}$, it is essential to infer some global consensus across nodes beforehand, thus making consistent actions in a decentralized manner. Indeed, the consensus inference lies in how to leverage limited available information $\textbf{z}_i^{(t)}$ to make the inferred state as close to $\textbf{s}^{(t)}$ as possible. Many solutions \cite{xiang2023decentralized, Tarmac2019, masia2022} have been proposed in this area. Considering its performance superiority, ConsMAC \cite{xiang2023decentralized} is taken into account in this manuscript, while the proposed solution is applicable to other works such as TarMAC \cite{Tarmac2019} and MASIA \cite{masia2022}.

Generally, ConsMAC first leverages the combination of a GRU-like memory module $\mathcal{F}_{\psi_M}$, parameterized by $\psi_M$, and positional encoding-based concatenation \cite{attention2017,gnn2020,gnn2024} to efficiently embed local observation and exchanged messages $\left\{\textbf{m}_j^{(t)} \mid j \in \xi_i^{(t)}\right\}$ in $\textbf{o}_i^{(t)}$.  Mathematically,
\begin{align}
&\mathbf{E}_{m_{i}}^{(t)} \label{Emi}
\\
=& \left[\mathcal{F}_{\psi_M} \left([\textbf{m}_{i_0}^{(t)} \, \| \, \textbf{o}_i^{(t)}]\right) \, \| \, \mathbf{\Phi}_{d_0}^{(t)}, \cdots, \left(\mathcal{A}_j \, \| \, \mathbf{\Phi}_{d_j}^{(t)} \right), \cdots\right]^\top,\nonumber
\end{align}
where $\left |  \right | $ denotes the concatenation operation,  
$
\mathbf{\Phi}_{d_j}^{(t)} = \sqrt{\frac{1}{D}} $$\left[\cos(w_1 L_2(\textbf{p}_{i_j}^{(t)})), \ldots, \cos(w_D L_2( \textbf{p}_{i_j}^{(t)} )) \right]^\top    \label{d}
$ with $D$ learnable weights $\psi_D = [w_1, \cdots, w_D]$, and $\mathcal{A}(\cdot)$ denotes the availability of information. Notably, for each agent $i \in \mathcal{N}$, $i_j$ represents the $j$-th nearest neighbor, while $i_0$ refers to the agent itself. Therefore, assuming the existence of a lossless channel from $i$ to $j$, $\mathcal{A}(\cdot)$ becomes valid only if the occurrence of sending messages $\textbf{m}_{i_j}^{(t)}$ from $j$ to $i$.

Afterward, each agent then aggregates a latent vector $\textbf{h}^{(t)}$ as
\begin{equation}
\begin{aligned}
\textbf{h}^{(t)} = \text{MHA}_{\psi_A} \left(\mathbf{E}_{m_i}^{(t)}, \mathbf{E}_{m_{i}}^{(t)}, \mathbf{E}_{m_{i}}^{(t)}\right),
\end{aligned}
    \label{H}
\end{equation}
where $\text{MHA}$ refers to a multi-head attention layer parameterized by $\psi_A$. On the basis, ConsMAC utilizes a global estimator $\mathcal{F}_{\psi_E}$, parameterized by $\psi_E$, to estimate the state embedding $\hat{\textbf{e}}^{(t)} = \mathcal{F}_{\psi_E}\left(\textbf{h}^{(t)}\right)$ in a supervised learning manner. Specifically, the Consensus Establishment (CE) loss function is computed as
\begin{equation}
\begin{aligned}
\mathcal{L}_{\text{CE}}(\Psi) = \mathbb{E}_{{t}}[\parallel\textbf{o}_g^{(t)}-\hat{\textbf{e}}^{(t)}\parallel^2],
\label{CE}
\end{aligned}
\end{equation}
where $\Psi = [\psi_D, \psi_M, \psi_A, \psi_E]$. 
Through this process, the intermediate outputs of ConsMAC implicitly encode the global information, effectively establishing a consensus that captures the states of all agents. Meanwhile, the information flow also serves as an optimal message representation that encompasses all necessary details for the GIB. Besides, the message $\textbf{m}_{i}^{(t+1)}$ for time-step $t+1$ will be computed as
\begin{equation}
    \textbf{m}_{i}^{(t+1)} = \mathbf{E}_{o_i}^{(t)} + \varpi _{i}^{(t)} \textbf{h}^{(t)},
    \label{mt+1}
\end{equation}
where the embedding vector $\mathbf{E}_{o_i}^{(t)}$ and the communication information weight \( \varpi _{i}^{(t)} \) can be obtained by Multi-Layer Percepton (MLP)-based encoders $\mathcal{F}_{\theta_O}$ and $\mathcal{F}_{\theta_W}$ as
$\mathbf{E}_{o_i}^{(t)} = \mathcal{F}_{\theta_O} \left( \textbf{o}_i^{(t)} \right)$ and $\varpi _i^{(t)} = \mathcal{F}_{\theta_W} \left( \textbf{o}_i^{(t)} \right)$.

Additionally, an executor $\mathcal{F}_{\theta_E}$ is employed to sample the final action output by calculating the mean of the Gaussian distribution as,
\begin{equation}
\begin{aligned}
\mu_i^{(t)} = \mathcal{F}_{\theta_E}\left(\textbf{m}_i^{(t+1)}\right), \quad \textbf{u}_i^{(t)} \sim \text{Normal}\left(\mu_i^{(t)}, \sigma^2\right),
\end{aligned}
    \label{action}
\end{equation}
where $\sigma^2$ represents the variance constant that introduces randomness to the agent's actions during exploration and gradually diminishes throughout the training process. As mentioned in Sec. \ref{mappo}, in alignment with MAPPO, we consider $[ \Theta, \Psi ]$ as the parameters of the final policy $\pi_{\theta_i}$ in Eq. \eqref{MAPPO}, where $\Theta = [\theta_O, \theta_W, \theta_E]$.

Based on the aforementioned consensus inference mechanism, the problem can be reformulated as the calibration of $\mathcal{A}$ by 
a variable-threshold event-triggering mechanism and global consensus inference $\mathbf{h}^{(t)}$. In other words,
if all agents share the same set of parameters during training to enhance learning efficiency, expressed as $\pi_{\theta_i} = \pi_{\theta}\ (\forall i \in \mathcal{N})$, it can be written as
\begin{equation}
\begin{aligned}
& \max_{\pi_{\theta}} \mathbb{E}_{t} \left[\sum_{t}\gamma ^{t}\mathcal{R}^{(t)} \mid \pi_{\theta}  \right], \\
\text{s.t.}\quad
&\textbf{u}_i(t) = \pi_\theta \left( g \left(\mathcal{A}^{(t)}, \textbf{h}^{(t)} \right) \right).
\end{aligned}
\end{equation}
where $g$ function denotes the specific calculations, which will be detailed in Sec. \ref{sec3_Framework}.

\section{The Framework of CDE-GIB}\label{sec3_Framework} 

The overall framework of the proposed CDE-GIB is shown in Fig. \ref{framework}. In addition to the policy execution module, it also encompasses the VT-ETM module and GIB-based consensus establishment module, which evaluate the information importance based on current and historic observations and encode concise message representation for consensus inference, respectively.

\begin{figure}[tp]
    \centering
    \includegraphics[scale=0.45]{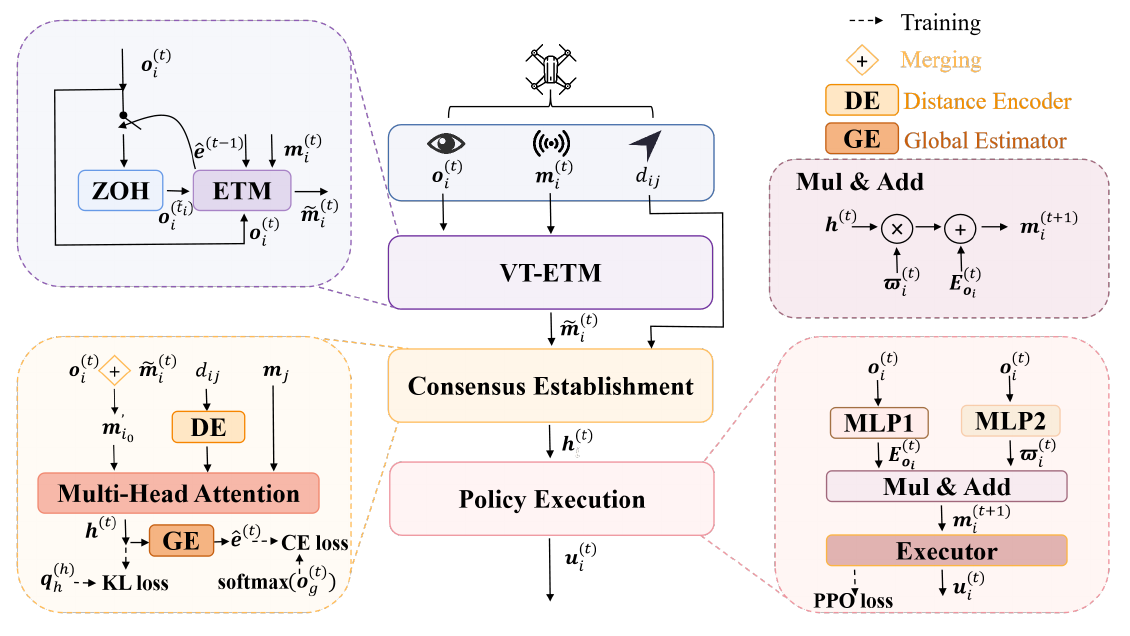}
    \vspace{-0cm}
\caption{The overall framework of CDE-GIB.} 
\vspace{-1.2em}
\label{framework}
\end{figure}

\subsection{Variable Threshold-Event Triggered Mechanism}\label{VT-ETMtitle}
The practical design of VT-ETM involves two parts. Firstly, to effectively reduce the frequency of event triggers, we use a zero vector $\textbf{o}_k^{(t)}$ as the label and apply Mean-Squared-Error (MSE) as the loss function to minimize the output of an event trigger function $\kappa(\cdot)$. Consequently, it ensures $\kappa(\cdot)$ to become as small as possible (i.e., the minimization of unnecessary triggering). Mathematically, the formula can be shown as
\begin{align}
&\mathcal{L}_{\text{VT-ETM}}(\psi_T) \label{eq:VT-ETM_loss} \\
= & \mathbb{E}_{{i, t}} \left[\parallel\kappa\left\{\mathcal{F}_{\psi_T}\left ( \textbf{o}_{i}^{\left ( t \right ) } , \textbf{o}_{i}^{\left ( \tilde{t_i}  \right ) } \right ) , \hat{\textbf{e}}^{(t-1)} \right\}- \textbf{o}_k^{(t)}\parallel^2 \right],\nonumber
\end{align}
where $\mathcal{F}_{\psi_T}$ denotes an MLP parameterized by $\psi_T$ to process the agent $i$'s observation $\textbf{o}_{i}^{\left ( t \right )}$ and its historical local observation $\textbf{o}_{i}^{\left( \tilde{t_i}  \right) }$ at the last trigger moment $\tilde{t_i}$ that is memorized by Zero-Order Hold (ZOH) \cite{ZOH}. 
Note that each agent $i$ maintains a sequence of event-triggering times $\mathcal{T}_i$ and when the agent determines that its information is semantically useful to other agents, it will add the corresponding time $t$ to the set $\mathcal{T}_i$. Correspondingly, $\tilde{t_i} = \displaystyle\arg \min_{\tau\in \mathcal{T}_i}{\left \{ t-\tau \right \}  }  \in \mathcal{T}_i$. Besides, the function $\kappa(\cdot)$ computes the similarity between the predictive information induced by $\mathcal{F}_{\psi_T}$ and the inferred consensus $\hat{\textbf{e}}^{(t-1)}$. Many similarity metrics can be adopted, such as cosine similarity, Manhattan distance, or Euclidean distance, depending on the characteristics of the data and the specific requirements of the task.
During the centralized training, we add $\psi_T$ to $\Psi = [\psi_D, \psi_M, \psi_A, \psi_E, \psi_T]$\footnote{For simplicity of representation, we slightly abuse the notations here.} for parameter updating.

On the other hand, during decentralized execution, we introduce an exponential function that decreases over time, namely $G_\text{threshold} = c\zeta  ^{t}$, where $c> 0$ and $0<\zeta  < 1$, as the variable threshold. In other words, for agent $i$
\begin{equation}
\mathcal{A}_i \text{ is }
\begin{cases}
\texttt{VALID}, &\!\!\!\! \kappa\left\{\mathcal{F}_{\psi_T}\left ( \textbf{o}_{i}^{\left ( t \right ) } , \textbf{o}_{i}^{\left ( \tilde{t_i}  \right ) } \right ) , \hat{\textbf{e}}^{(t-1)}\right\} > G_{\text{threshold}}; \\
\texttt{VOID}, &\!\!\!\! \kappa\left\{\mathcal{F}_{\psi_T}\left ( \textbf{o}_{i}^{\left ( t \right ) } , \textbf{o}_{i}^{\left ( \tilde{t_i}  \right ) } \right ) , \hat{\textbf{e}}^{(t-1)}\right\} \leq G_{\text{threshold}}.
\end{cases}
    \label{VT-ETM}
\end{equation}
Exploiting a time-decreasing threshold $G_\text{threshold}$ aligns with the branch-out approaches used to mitigate accumulative errors during model rollouts \cite{rollout}, 
progressively compromising the accuracy of predictions. Therefore, an adaptive threshold that encourages increasingly frequent updates can naturally counteract the gradually enlarged errors. 
Finally, if $\mathcal{A}_i$ is \texttt{VOID}, neighboring agents will be unable to receive any exchanged messages from agent $i$. In this case, instead of using the memorized message in \cite{ETCNET2023}, the corresponding elements in Eq. \eqref{Emi} will be replaced by an all-zero vector. Such a setting contributes to allow agents receiving messages infrequently to avoid over-reliance on outdated memorized messages that may hinder their ability to make timely decisions.

\subsection{Graph Information Bottleneck for the Communication Graph and Information Flow Optimization}
As mentioned above, MAGI \cite{MAGI2024} generates excessive communication volume due to the cumbersome interaction with all agents within a specified range. Moreover, the algorithm suffers from high computational complexity caused by the separate compression of the communication graph and data flow. To address this challenge, we propose the GIB-based joint optimization. Specifically, different from MAGI \cite{MAGI2024}, which considers the triplet $\langle$\texttt{Feature}, \texttt{Communication/Graph Information}, \texttt{Explicit Action}$\rangle$, our method incorporates a novel triplet $\langle$\texttt{Feature}, \texttt{Implicit Consensus}, \texttt{Global Observation}$\rangle$ to infer the consensus from data flow and the communication graph simultaneously.

Beforehand, for graph-structured agent features, we introduce a local-dependence assumption to avoid explicitly requiring input data to be i.i.d.

\begin{Assumption}
    For each agent $i$, given the neighbor-related agents within a certain number of hops, the features of the remaining agents are considered independent of the feature of agent $i$. 
\label{assumption}
\end{Assumption}

Contingent on Assumption \ref{assumption}, we discuss the implementations of the GIB-based joint compression of the communication graph and information flow towards a more compact consensus representation. Without loss of generality, the input feature data for the communication learning mechanism based on GNNs can be universally expressed as $\mathcal{D} = (\mathcal{E}, \mathcal{H})$, where $\mathcal{D}$ denotes the aggregated message and the embedding representation $\textbf{E}_{m_i}^{(t)}$ as shown in \eqref{Emi}. Therefore, our primary objective is to compress the consensus $\textbf{h}^{(t)}$ from $\textbf{E}_{m_i}^{(t)}$, while promoting consensus $\textbf{h}^{(t)}$ to closely approximate the target global observation labels $\textbf{o}_g^{(t)}$. Mathematically, we want to minimize $\mathcal{L}_{\text{IB}} = -I(\textbf{h}^{(t)}; \textbf{o}_g^{(t)}) + \eta I(\textbf{h}^{(t)}; \textbf{E}_{m_i}^{(t)})$. Due to the difficulty to know the joint distribution $p(\textbf{h}^{(t)}, \textbf{o}_g^{(t)})$ and $p(\textbf{h}^{(t)}, \textbf{E}_{m_i}^{(t)})$, by Lemma \ref{I_NWJ1}, we have the following theorem.

\begin{lemma}[Nguyen, Wainright \& Jordan’s bound \cite{GIB}]
\label{I_NWJ1}
For two random variables $X$ and $Y$, 
\begin{align}
I\left(Y ; X\right) &\geq 1+\mathbb{E}_{p(Y)}\left[\log \frac{\prod_{i \in \mathcal{N}} p_{1}\left(Y_i \mid X_{i}\right)}{p_{2}(Y)}\right]\\
&-\mathbb{E}_{p(Y) p\left(X\right)}\left[\frac{\prod_{i \in \mathcal{N}} p_{1}\left(Y_{i} \mid X_i\right)}{p_{2}(Y)}\right].\nonumber
\end{align} 
\end{lemma}

\begin{theorem}
\label{theorem:gib}
 The GIB can be bounded by
    \begin{align}
    \mathcal{L}_{\text{IB}} \leq& \mathbb{E}_{{i, t}} \left[ D_{\text{KL}} \left( p(\textbf{o}_g^{(t)}) \parallel p(\textbf{h}^{(t)}) \right) \right] \label{IB}\\
     &\underbrace{+ \eta\mathbb{E}_{p(\textbf{E}_{m_i}^{(t)})} \left[ D_{\text{KL}} \left( p(\textbf{h}^{(t)} \mid \textbf{E}_{m_i}^{(t)}) \parallel q_{\textbf{h}}(\textbf{h}^{(t)}) \right) \right]}_{\mathcal{L}_\text{KL}}, \nonumber
    \end{align}
    where $q_{\textbf{h}}(\textbf{h}^{(t)}) $ denotes a probability function sharing the same variable space as $p(\textbf{h}^{(t)} \mid \textbf{E}_{m_i}^{(t)})$.
\end{theorem}
\begin{proof}
    \begin{align}
    &I\left(\textbf{h}^{(t)}; \textbf{o}_g^{(t)}\right) \nonumber                   \\
    \overset{(a)}{\geq}&1+\mathbb{E}_{{p}(\textbf{o}_g^{(t)})}\left[\log \frac{ p(\mathcal{F}_{\psi_E}\left(\textbf{h}^{(t)}\right))}{p\left(\textbf{o}_g^{(t)}\right)}\right]      \nonumber \\
    &\quad -\mathbb{E}_{p{(\textbf{o}_g^{(t)})}p{(\textbf{h}^{(t)})}}\left[\frac{ p(\mathcal{F}_{\psi_E}\left(\textbf{h}^{(t)}\right))}{p\left(\textbf{o}_g^{(t)}\right)}\right]    \label{lower_IB}     \\
    \overset{(b)}{\geq} &- \mathbb{E}_{{p}(\textbf{o}_g^{(t)})}\left[\log \frac{p\left(\textbf{o}_g^{(t)}\right)}{p(\mathcal{F}_{\psi_E}\left(\textbf{h}^{(t)}\right))}\right]    \nonumber   \\
    =&- \mathbb{E}_{{t}} \left[ D_{\text{KL}} \left( p(\textbf{o}_g^{(t)}) \parallel p(\mathcal{F}_{\psi_E}(\textbf{h}^{(t)})) \right) \right], \nonumber
    \end{align}
    where the inequality (a) uses the Nguyen, Wainright \& Jordan’s bound $I_\text{NWJ}$ \cite{GIB} and the inequality (b) is derived from the condition $1-\mathbb{E}_{p{(\textbf{o}_g^{(t)})}p{(\textbf{h}^{(t)})}}\left[\frac{ p(\mathcal{F}_{\psi_E}\left(\textbf{h}^{(t)}\right))}{p\left(\textbf{o}_g^{(t)}\right)}\right] > 0$. 
    On the other hand,
    \begin{align}
    &I(\textbf{h}^{(t)}; \textbf{E}_{m_i}^{(t)}) \nonumber \\
    =&\iint p(\textbf{h}^{(t)},\textbf{E}_{m_i}^{(t)})\log \frac{p(\textbf{h}^{(t)},\textbf{E}_{m_i}^{(t)})}{p(\textbf{h}^{(t)})p(\textbf{E}_{m_i}^{(t)})} \, d\textbf{h}^{(t)} \, d\textbf{E}_{m_i}^{(t)}  \label{upper_IB}  \\
    =& \iint p(\textbf{E}_{m_i}^{(t)}) p(\textbf{h}^{(t)} \mid \textbf{E}_{m_i}^{(t)}) \log \frac{p(\textbf{h}^{(t)} \mid \textbf{E}_{m_i}^{(t)})}{p(\textbf{h}^{(t)})} \, d\textbf{h}^{(t)} \, d\textbf{E}_{m_i}^{(t)} \nonumber \\
    \overset{(c)}{\leq}& \iint p(\textbf{E}_{m_i}^{(t)}) p(\textbf{h}^{(t)} \mid \textbf{E}_{m_i}^{(t)}) \log \frac{p(\textbf{h}^{(t)} \mid \textbf{E}_{m_i}^{(t)})}{q_{\textbf{h}}(\textbf{h}^{(t)})} \, d\textbf{h}^{(t)} \, d\textbf{E}_{m_i}^{(t)} \nonumber \\
    =& \mathbb{E}_{p(\textbf{E}_{m_i}^{(t)})} \left[ D_{\text{KL}} \left( p(\textbf{h}^{(t)} \mid \textbf{E}_{m_i}^{(t)}) \parallel q_{\textbf{h}}(\textbf{h}^{(t)}) \right) \right],\nonumber 
    \end{align} 
where the inequality (c) originates from Gibbs' inequality, where $p \log p \geq p \log q$, with equality if and only if $p$ and $q$ are the same distribution. To further estimate $p(\textbf{h}^{(t)})$, we treat $q_{\textbf{h}}(\textbf{h}^{(t)})$ as the variational approximation.

\end{proof}
Compared to Eq. \eqref{CE}, Theorem \ref{theorem:gib} unveils the impact of GIB on the consensus-building module. In a nutshell, the loss function of CDE-GIB can be expressed as,
\begin{equation}
\mathcal{L}_{\text{CDE-IB}}(\Theta, \phi, \Psi) = -J_{\text{MAPPO}} + \mathcal{L}_{\text{CE}} + \varrho\mathcal{L}_{\text{VT-ETM}} + \rho\mathcal{L}_{\text{GIB}}.
    \label{loss}
\end{equation}

However, it remains difficult to compute the KL divergence directly during training. Fortunately, we have the following lemma. 

\begin{lemma}
[Ref. \cite{GMM}]
    \label{lem:guassian_kl}
    Considering $f(\mathbf{x})$ and $g(\mathbf{x})$ ($\mathbf{x} \in \mathbb{R}^K$) as Gaussian distributions, that is,
    \begin{equation}
        f(\mathbf{x}) = \mathcal{N}(\bm{\mu}_{m}, \bm{\Sigma}_{m}),\ 
        g(\mathbf{x}) =  \mathcal{N}(\bm{\mu}_{l}, \bm{\Sigma}_{l}),
    \end{equation}
    where $\bm{\mu} = [\mu^1,\cdots,\mu^K]$\footnote{For simplicity of representation, we omit the subscript $m$ and $l$ for $\bm{\mu}$ and $\bm{\Sigma}$ here.} and $\bm{\Sigma} = \text{diag}[(\sigma^1)^2,\cdots,(\sigma^K)^2]$, 
    the KL divergence between these two distributions can be computed as
    \begin{align}
        &\mathbb{E} \left[ D_{\text{KL}} \left( f(\mathbf{x}) \parallel g(\mathbf{x}) \right) \right] \label{eq:kl_guassian_guassian}\\ =&\mathbb{E}\left[ \sum_{k=1}^{K} \left(\log \frac{\sigma_{m}^k}{\sigma_{l}^k}+\frac{\left(\sigma_{m}^k\right)^{2}+\left(\mu_{m}^k-\mu_{l}^k\right)^{2}}{2\left(\sigma_{l}^k\right)^{2}}-\frac{1}{2}\right)\right].\nonumber
    \end{align}
\end{lemma}
Contingent on the following Assumption \ref{assump:guassian}, we can have a corollary to faciliate the computations of GIB, which lays the very foundation for computing GIB across batches. 
\begin{Assumption}[Consistent with Ref. \cite{TOCF}]
Taking a batch of collected data, we assume $\textbf{o}_g^{(t)}$, $\mathcal{F}_{\psi_E}(\textbf{h}^{(t)})$, and $\textbf{h}^{(t)} $ satisfy the following Gaussian distributions,
\begin{align}
\label{eq:guassian}
    p(\textbf{o}_g^{(t)}) =& \mathcal{N}(\bm{\mu}_{\textbf{o}_g^{(t)}},\bm{\Sigma}_{\textbf{o}_g^{(t)}}), \nonumber\\ 
   p(\mathcal{F}_{\psi_E}(\textbf{h}^{(t)})) =& \mathcal{N}(\bm{\mu}_{\mathcal{F}_{\psi_E}(\textbf{h}^{(t)})},\bm{\Sigma}_{\mathcal{F}_{\psi_E}(\textbf{h}^{(t)})}),\\
   p(\textbf{h}^{(t)} ) =& \mathcal{N}(\bm{\mu}_{\textbf{h}^{(t)} \mid \textbf{E}_{m_i}^{(t)}},\bm{\Sigma}_{\textbf{h}^{(t)} \mid \textbf{E}_{m_i}^{(t)}}),\nonumber
\end{align}
where 
$\bm{\mu}_{\textbf{o}_g^{(t)}},\ \bm{\mu}_{\mathcal{F}_{\psi_E}(\textbf{h}^{(t)})},\ \bm{\sigma}_{\textbf{o}_g^{(t)}},\ \bm{\sigma}_{\mathcal{F}_{\psi_E}(\textbf{h}^{(t)})} \in \mathbb{R}^{K_1}$, and 
$\bm{\mu}_{\textbf{h}^{(t)} \mid \textbf{E}_{m_i}^{(t)}},\ \bm{\sigma}_{\textbf{h}^{(t)} \mid \textbf{E}_{m_i}^{(t)}} \in \mathbb{R}^{K_2}$.

\label{assump:guassian}
\end{Assumption}

\begin{corollary}
\label{coro:gib}
    Contingent on Assumption \ref{assump:guassian}, if the variational approximation $q_{\textbf{h}}(\textbf{h}^{(t)}) = \mathcal{N}(\mathbf{0}, \mathbf{I})$, the GIB bound can be computed as Eq. \eqref{eq:gib_opt}.
\end{corollary}
The corollary can be easily obtained by applying Assumption \ref{assump:guassian} and Lemma \ref{lem:guassian_kl} in Theorem \ref{theorem:gib}. In summary, the training procedure for CDE-GIB is presented in Algorithm \ref{algorithm}.

\begin{figure*}
\begin{equation}
\begin{aligned}
    \mathcal{L}_{\text{GIB}}
    \leq&\mathbb{E}\left[ \sum_{k=1}^{K_1} \left(\log \frac{\sigma_{\textbf{o}_g^{(t)}}^k}{\sigma_{\mathcal{F}_{\psi_E}(\textbf{h}^{(t)})}^k}+ \frac{\left(\sigma_{\textbf{o}_g^{(t)}}^k\right)^{2}+\left(\mu_{\textbf{o}_g^{(t)}}^k-\mu_{\mathcal{F}_{\psi_E}(\textbf{h}^{(t)})}^k\right)^{2}}{2\left(\sigma_{\mathcal{F}_{\psi_E}(\textbf{h}^{(t)})}^k\right)^{2}}\right)\right] \\
    &+\mathbb{E}\left[ \sum_{k=1}^{K_2} \left(\log \sigma_{\textbf{h}^{(t)} \mid \textbf{E}_{m_i}^{(t)}}^k
    +\frac{\left(\sigma_{\textbf{h}^{(t)} \mid \textbf{E}_{m_i}^{(t)}}^k\right)^{2}+\left(\mu_{\textbf{h}^{(t)} \mid \textbf{E}_{m_i}^{(t)}}^k\right)^{2}}{2}\right)\right]-1.
\end{aligned}\label{eq:gib_opt}
\end{equation}
\hrulefill
\vspace{-1.5em}
\end{figure*}

\newcommand{\INITIALIZE}{\item[\textbf{Initialize:}]}
\addtolength{\topmargin}{0.05in}
\begin{algorithm}[tbp]\small
  \caption{The Training of CDE-GIB}
  \label{algorithm}
  \begin{algorithmic}[1]
    \INITIALIZE {The length of episodes $T$, variance constant $\sigma$, the actor and critic network with random parameters $\Theta, \Psi, \phi$ and the replay memory $\mathcal{B} \leftarrow \varnothing$;}
    \FOR {each train epoch}
        \STATE $\text{Clone}$ $\Theta_{\text{old}} \leftarrow \Theta$, $\Psi_{\text{old}} \leftarrow \Psi$, $\phi_{\text{old}} \leftarrow \phi$;
        \STATE Initialize the environment with $N$ agents;
        \FOR {$t=\{1, \cdots, T\}$}
            \FOR {\textbf{each agent $i$}}
            \STATE $\textbf{z}_i^{(t)} \leftarrow $ obtains a local state $\textbf{z}_i^{(t)} =\left(\textbf{o}_i^{(t)}, \left\{\textbf{m}_j^{(t)} \mid j \in \xi_i^{(t)}\right\}\right)$;
            \STATE $\mathcal{A}_i,\textbf{E}_{m_{i}}^{(t)} \leftarrow $ calculates triggering behavior $\mathcal{A}_i$ by Eq. \eqref{VT-ETM} and generates encoded information $\mathbf{E}_{m_{i}}^{(t)}$ by Eq. \eqref{Emi};
            \STATE $\textbf{h}^{(t)}, \textbf{m}_i^{(t+1)}, \textbf{u}_i^{(t)} \leftarrow $ establishes the consensus $\textbf{h}^{(t)}$ by Eq. \eqref{H} and computes the message $\textbf{m}_i^{(t+1)}$ by Eq. \eqref{mt+1} to sample an action $\textbf{u}_i^{(t)} \sim \text{Normal}(\mu^{(t)}_{\text{i, old}}, \sigma^2)$ by Eq. \eqref{action};
            \ENDFOR
        \STATE $r^{(t)} , V_{\phi_{old}}(\textbf{s}^{(t)}) ,\textbf{s}^{(t+1)}\leftarrow $ obtain the reward $r^{(t)}$, state value $V_{\phi_{old}}(\textbf{s}^{(t)})$ and $\textbf{s}^{(t+1)}$;
        \ENDFOR
        \STATE For each time-step $t$, each agent calculates $\mu^{(t)}_i$, $\hat{\textbf{e}}^{(t)}$ based on $\textbf{z}^{(t)}_i$ by Eq. \eqref{Emi}-\eqref{IB}, and obtains $V_{\phi}(\textbf{s}^{(t)})$;
        \STATE Update $\Theta, \Psi, \phi$ according to Eq. \eqref{loss} via Adam optimizer;
    \ENDFOR
  \end{algorithmic}
\end{algorithm}
\section{Simulation Settings and Results}\label{sec4_Experiment}
\subsection{Simulation Settings}
In this section, we evaluate the performance of CDE-GIB in terms of executing the decentralized formation control task \cite{xiang2023decentralized} in the multi-agent particle environment \cite{MPE2017}. Notably, we use cosine similarity during computing $\kappa$ in Eq. \eqref{eq:VT-ETM_loss}. 
Additionally, the reward function $\mathcal{R}^{(t)}$ is defined as a weighted combination of task-oriented reward $\mathcal{R}_t^{(t)}$ and the event-triggered reward $\mathcal{R}_m^{(t)}$. Specifically, consistent with \cite{xiang2023decentralized}, the task-oriented reward $\mathcal{R}_t^{(t)}$, which contains the formation completeness, individual navigation distance, and penalty on collision arising from decentralized control, evaluates the agents' efficacy in executing specified tasks. Meanwhile, the event-triggered reward $\mathcal{R}_m^{(t)}$ imposes a penalty on agents for transmitting information. 
Mathematically, 
\begin{equation}
    \mathcal{R}^{(t)} = \omega_k \mathcal{R}_t^{(t)} + \omega_m \mathcal{R}_m^{(t)},
\end{equation}
where the coefficients $ \omega_t, \omega_m $ represent the corresponding weights. We also evaluate the distributed consensus establishment method ConsMAC\cite{xiang2023decentralized}, the attention-based message aggregation approach TarMAC \cite{Tarmac2019}, and the state-of-the-art supervised learning-based information extraction algorithm MASIA\cite{masia2022} as baselines. The key parameters are summarized in Table \ref{para}.

\begin{table}[t!]
    \centering
    \caption{The Key Parameter Settings of The Environment.}
    \vspace{-.5em}
    \label{para}
    \begin{tabular}{@{}c|c|c@{}}
        \toprule
        \textbf{Environment Parameters} & \textbf{Symbol}& \textbf{Value} \\ 
        \midrule
        Number of UAVs & $N$ & $7$ \\
        Maximum observation distance & $\delta _{obs}$ & $3$ m \\
        Destination & $\Delta_p$ & $(0, 10)$ m\\
        Discount factor & $\gamma $ & $0.8$ \\
        GAE factor & $\lambda $ & $0.95$\\
        The range of acceleration & $\textbf{u}_i^{(t)}$ & $[-0.5, 0.5] $ m/s\textsuperscript{2}\\
        The range of position & $\textbf{p}_i^{(0)},\textbf{p}_j^{(0)}$ & $[-2, 2]$ m\\
        Reward function coefficients & $\omega_k,\omega_m$ &  $1$, $0.1$ \\
        \bottomrule
    \end{tabular}
    \vspace{-1.5em}
\end{table}
\subsection{Simulation Results}
\begin{figure}[tp]
    \centering
    \includegraphics[scale=0.45]{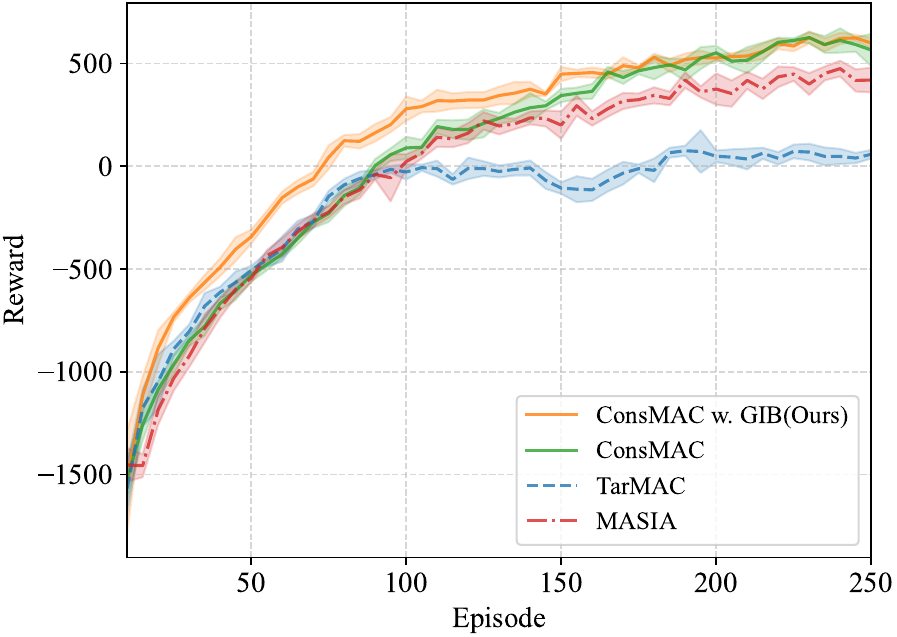}
    \vspace{-0.2cm}
\caption{Learning curves of consensus algorithms with and without GIB optimization.} 
\vspace{-1em}
\label{experiment2}
\end{figure} 
We first conduct ablation studies to show the contribution of individual modules. As shown in Table \ref{communication_volume_m} and \ref{communication_volume_h}, the comparison between ConsMAC with and without ETM indicates the removal of the VT-ETM component causes a slight decline in information processing performance and a significant increase in redundant communication volume $\tilde{\textbf{m}}_{i}^{(t)}$. Moreover, removing the GIB module also results in a compression performance degradation and a greater accumulation of unnecessary inferred consensus $\textbf{h}^{(t)}$. Afterward, we compare ConsMAC with GIB on top of ConsMAC with other baselines  \cite{xiang2023decentralized,Tarmac2019,masia2022}. Fig. \ref{experiment2} presents the corresponding result. It can be observed that due to the incorporation of GIB, it significantly outperforms other baselines, including the state-of-the-art ConsMAC algorithm.
\begin{table}[t!]
    \centering
    \caption{Comparison of the communication volume with and without VT-ETM, under different maximum observation ranges $\delta _\text{com}$.}
    \vspace{-.5em}
    \begin{tabular}{ccccccccc}
        \toprule
        $\delta _\text{com}$ (m) &$\mathbf{2.1}$ & $\mathbf{2.4}$ & $\mathbf{2.7}$ & $\mathbf{2.8}$ \\ 
        \midrule
        w. ETM  & $\mathbf{846.32}$ & $\mathbf{880.47}$ & $\mathbf{994.17}$ & $\mathbf{931.34}$ \\
        w.o. ETM & $874.95$ & $950.84$ & $1086.77$ & $955.89$ \\
        \bottomrule
    \end{tabular}
    \label{communication_volume_m}
    \vspace{-2em}
\end{table}

\begin{table}[t!]
    \centering
    \caption{Comparison of inferred consensus with and without GIB, under different maximum observation ranges $\delta _\text{com}$.}
    \vspace{-.5em}
    \begin{tabular}{ccccccccc}
        \toprule
        $\delta _\text{com}$ (m) &$\mathbf{2.1}$ & $\mathbf{2.4}$ & $\mathbf{2.7}$ & $\mathbf{2.8}$ \\ 
        \midrule
        w. GIB  & $\mathbf{208.35}$ & $\mathbf{199.77}$ & $\mathbf{216.48}$ & $\mathbf{202.54}$ \\
        w.o. GIB & $450.38$ & $464.04$ & $506.77$ & $511.13$ \\
        \bottomrule
    \end{tabular}
    \label{communication_volume_h}
    \vspace{-1em}
\end{table}



Finally, we focus on the universal applicability of VT-ETM plugin in other baselines, and provide the related results in Fig. \ref{experiment3}. A clear trend emerges that the adoption of VT-ETM leads to notable performance improvement for all methods. In particular, 
integrating the VT-ETM mechanism into the ConsMAC algorithm yields an even more remarkable performance enhancement, further reinforcing the effectiveness and reliability of the ETM plugin across diverse algorithms.

\begin{figure}[tp]
    \vspace{-0cm}
    \centering
    \includegraphics[scale=0.45]{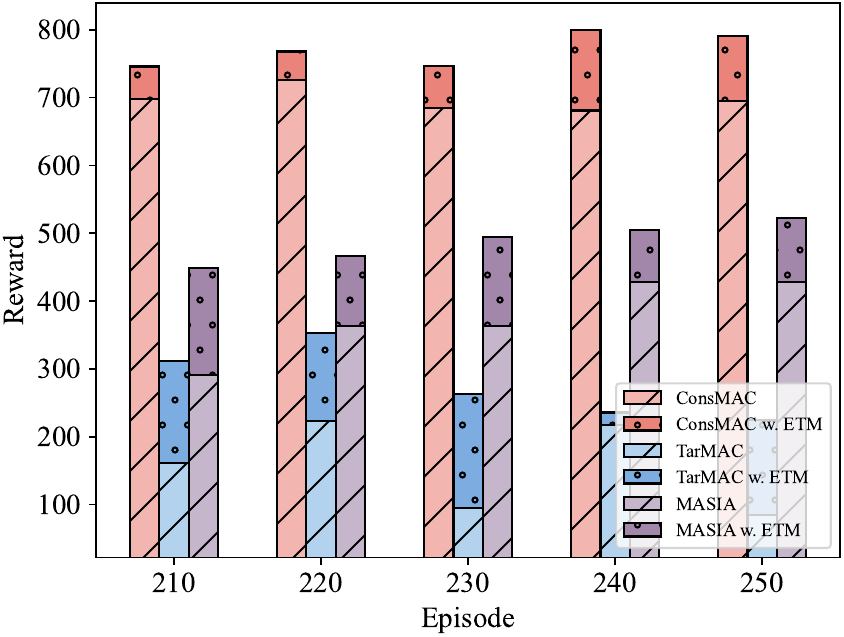}
    \vspace{-0.2cm}
\caption{Performance Comparison of consensus algorithms with and without ETM.} 
\label{fig:R} 
\vspace{-1.6em}
\label{experiment3}
\end{figure} 

\section{Conclusions}\label{sec5_Conclusions}
In this work, we have proposed and validated CDE-GIB, a robust, event-triggered integrated communication and control framework with GIB optimization. To be specific, 
we have implemented a GIB module that jointly optimizes the communication graph and data flow in ConsMAC methodology, which effectively compresses the consensus into a sufficient and compact representation. Additionally, a VT-ETM algorithm has been employed to assess the information importance based on the fusion of historical data and current observations, while an opportunistic transmission mechanism has been leveraged to reduce the dissemination of redundant communication messages during the interactive process of reaching consensus. We have conducted extensive experiments to demonstrate the effectiveness and adaptability of our proposed method in communication-limited environments. In future work, we will further explore larger-scale formations under stricter communication constraints and deploy the approach on a more practical hardware platform.
\bibliographystyle{IEEEtran}
\bibliography{arxiv}
\end{document}